\documentclass[fleqn]{llncs}
\usepackage{ttc}
\usepackage{eurosym}
\usepackage{version}

\title{Timed Tuplix Calculus and\\
       the Wesseling and van den Bergh Equation%
       \thanks{This research was carried out in the framework of
               the  Jacquard-project Symbiosis, which is funded by the
               Netherlands Organisation for Scientific Research (NWO).}}
\author{J.A. Bergstra \and C.A. Middelburg}
\institute{Informatics Institute, Faculty of Science,
           University of Amsterdam \\
           Science Park~904, 1098~XH Amsterdam, the Netherlands \\
           \email{J.A.Bergstra@uva.nl,C.A.Middelburg@uva.nl}}

\begin{document}

\maketitle

\begin{abstract}
% 102 %
We develop an algebraic framework for the description and analysis of 
financial behaviours, that is, behaviours that consist of transferring 
certain amounts of money at planned times.
To a large extent, analysis of financial products amounts to analysis of 
such behaviours.
We formalize the cumulative interest compliant conservation requirement
for financial products proposed by Wesseling and van den Bergh by an 
equation in the framework developed and define a notion of financial 
product behaviour using this formalization.
We also present some properties of financial product behaviours.
The development of the framework has been influenced by previous work on
the process algebra ACP.
% \begin{classcode}
% % MSC 2000: ...
% \end{classcode}
\end{abstract}

\section{Introduction}
\label{sect-introduction}

Analysis of financial products amounts to a large extent to analysis of
behaviours that consist of transferring certain amounts of money at 
planned times.
In this paper, such behaviours are called financial behaviours.
Mathematically precise analysis of financial products is complicated by 
the lack of a specialized mathematical framework for the description and 
analysis of financial behaviours.
The main objective of the work presented in this paper is to devise such 
a framework.
We aim at an algebraic framework, that is, a framework in which 
operators enable us to describe a financial behaviour as a behaviour 
composed of several other financial behaviours and equational axioms
enable us to analyze a described financial behaviour by mere algebraic
calculations.
Our intuitive understanding of the nature of financial behaviours will
provide the primary justification of the equations that are taken as 
axioms.

To achieve our main objective, we develop an extension of the core of 
tuplix calculus that can deal with the timing of transfers involved in 
financial behaviours.
Tuplix calculus was presented for the first time in~\cite{BPZ07a} and
has among other things been applied in modular financial budget design.
The extension of the core of tuplix calculus developed in this paper is 
called timed tuplix calculus.
The operators added to the core of tuplix calculus in this extension are
comparable to operators introduced earlier in the setting of the process
algebra ACP~\cite{BK84b}.
In the core of tuplix calculus as well as timed tuplix calculus, the 
mathematical structure for quantities is a signed cancellation 
meadow~\cite{BBP13a}.
The prime examples of cancellation meadows are the fields of rational
and real numbers with the multiplicative inverse operation made total by 
imposing that the multiplicative inverse of zero is zero.
A cancellation meadow is an appropriate mathematical structure for
quantities.
A signed cancellation meadow is a cancellation meadow expanded with a
signum operation.

In~\cite{WB00a}, Wesseling and van den Bergh formulate a cumulative 
interest compliant conservation requirement for financial products: 
the sum of all transfers relating to the product, transposed to some 
point of time (the focal date) by means of cumulative interest at the 
effective interest rate of the product, is zero.
As an example of the use of timed tuplix calculus, we formalize this 
conservation requirement by an equation in timed tuplix calculus.
Unaware of previous occurrences of the requirement in the financial
literature, we call this equation the Wesseling and van den Bergh
equation.
Using this equation, we define a notion of financial product behaviour.
A financial product behaviour can be seen as a financial behaviour for 
which a financial product can be devised that involves that behaviour.

In addition to that, we adapt the notion of implicit capital of a 
process introduced in~\cite{BM06e} to the current setting.
The implicit capital associated with a financial behaviour can be seen
as the least amount of money that must be at disposal initially to 
exhibit that behaviour, taking cumulative interest into account.
We use this notion to show that financial behaviours may profit from 
using some financial product.
We also present some other properties of financial product behaviours.

This paper is organized as follows.
First, we give a brief summary of signed cancellation meadows
(Section~\ref{sect-canc-meadows}).
Next, we review the core of tuplix calculus (Section~\ref{sect-CTC}).
Then, we extend the core of tuplix calculus to timed tuplix calculus
(Section~\ref{sect-TTC}).
After that, we formalize the conservation requirement for financial
products, define a notion of financial product behaviour, and present 
some properties of financial product behaviours 
(Section~\ref{sect-fpb}).
Following this, we construct the standard model of the timed tuplix
calculus (Section~\ref{sect-model}).
Finally, we make some concluding remarks
(Section~\ref{sect-conclusions}).

\section{Signed Cancellation Meadows}
\label{sect-canc-meadows}

In the timed tuplix calculus presented in this paper, the mathematical
structure for quantities is a signed cancellation meadow.
In this section, we give a brief summary of signed cancellation meadows.

A meadow is a field with the multiplicative inverse operation made total
by imposing that the multiplicative inverse of zero is zero.
A cancellation meadow is a meadow in which the multiplicative inverse
operation satisfies the general inverse law (given below).
A signed meadow is a meadow expanded with a signum operation.
Meadows were defined for the first time in~\cite{BT07a} and elaborated
in several subsequent papers.
The expansion of meadows with a signum operation originates
from~\cite{BBP13a}.
In the latter paper, references are made to the key papers on meadows.

The signature of meadows consists of the following constants and
operators:
\begin{itemize}
\item
the constants $0$ and $1$;
\item
the binary \emph{addition} operator ${} +$ {};
\item
the binary \emph{multiplication} operator ${} \mul {}$;
\item
the unary \emph{additive inverse} operator $- {}$;
\item
the unary \emph{multiplicative inverse} operator ${}\minv$.
\end{itemize}

We assume that there are infinitely many variables, including $u$, $v$
and $w$.
Terms are build as usual.
We use infix notation for the binary operators \mbox{${} + {}$} and
\mbox{${} \mul {}$}, prefix notation for the unary operator 
\mbox{$- {}$}, and postfix notation for the unary operator 
\mbox{${}\minv$}.
We use the usual precedence convention to reduce the need for
parentheses.
We introduce subtraction and division as abbreviations:
$p - q$ abbreviates $p + (-q)$ and
$p / q$ abbreviates $p \mul q\minv$.
We use numerals in the common way ($2$ abbreviates $1 + 1$, etc.).
We also use the notation $p^n$ for exponentiation with a natural number
as exponent.
For each term $p$ over the signature of meadows, the term $p^n$ is
defined by induction on $n$ as follows: $p^0 = 1$ and
$p^{n+1} = p^n \mul p$.

The constants and operators from the signature of meadows are adopted
from rational arithmetic, which gives an appropriate intuition about
these constants and operators.

A meadow is an algebra over the signature of meadows that satisfies the
equations given in Table~\ref{eqns-meadow}.%
\begin{table}[!t]
\caption{Equations for meadows}
\label{eqns-meadow}
\begin{eqntbl}
\begin{eqncol}
(u + v) + w = u + (v + w)                                             \\
u + v = v + u                                                         \\
u + 0 = u                                                             \\
u + (-u) = 0
\end{eqncol}
\qquad\quad
\begin{eqncol}
(u \mul v) \mul w = u \mul (v \mul w)                                 \\
u \mul v = v \mul u                                                   \\
u \mul 1 = u                                                          \\
u \mul (v + w) = u \mul v + u \mul w
\end{eqncol}
\qquad\quad
\begin{eqncol}
(u\minv)\minv = u                                                   \\
u \mul (u \mul u\minv) = u
\end{eqncol}
\end{eqntbl}
\end{table}
Thus, a meadow is a commutative ring with identity equipped with a
multiplicative inverse operation \mbox{${}\minv$} satisfying the 
reflexivity equation ${(u\minv)}\minv = u $ and the restricted inverse 
equation $u \mul (u \mul u\minv) = u$.
From the equations given in Table~\ref{eqns-meadow}, the equation
$0\minv = 0$ is derivable (see~\cite{BT07a}).

In meadows, the multiplicative inverse operation is total.
The advantage of working with a total multiplicative inverse operation
lies in the fact that conditions like $u \neq 0$ in
$u \neq 0 \Implies u \mul u\minv = 1$ are not needed to guarantee
meaning.

A \emph{non-trivial meadow} is a meadow that satisfies the
\emph{separation axiom}
\begin{ldispl}
0 \neq 1\;;
\end{ldispl}
and a \emph{cancellation meadow} is a meadow that satisfies the
\emph{cancellation axiom}
\begin{ldispl}
u \neq 0 \And u \mul v = u \mul w \Implies v = w\;,
\end{ldispl}
or equivalently, the \emph{general inverse law}
\begin{ldispl}
u \neq 0 \Implies u \mul u\minv = 1\;.
\end{ldispl}
Important properties of non-trivial cancellation meadows are 
$u / u = 0 \Iff u = 0$ and $u / u = 1 \Iff u \neq 0$.

A \emph{signed meadow} is a meadow expanded with a unary \emph{signum}
operation $\sign$ satisfying the equations given in
Table~\ref{eqns-signum}.%
\begin{table}[!t]
\caption{Equations for signum operation}
\label{eqns-signum}
\begin{eqntbl}
\begin{eqncol}
\sign(u / u) = u / u                                                  \\
\sign(1 - u / u) = 1 - u / u                                          \\
\sign(-1) = -1
\end{eqncol}
\qquad\quad
\begin{eqncol}
\sign(u\minv) = \sign(u)                                              \\
\sign(u \mul v) = \sign(u) \mul \sign(v)                              \\
(1 - \frac{\sign(u) - \sign(v)}{\sign(u) - \sign(v)}) \mul
(\sign(u + v) - \sign(u)) = 0
\end{eqncol}
\end{eqntbl}
\end{table}
In combination with the cancellation axiom, the last equation in this
table is equivalent to the conditional equation
$\sign(u) = \sign(v) \Implies \sign(u + v) = \sign(u)$.

In signed cancellation meadows, the function $\max$ is defined as 
follows:
\begin{ldispl}
\max(u,v) =
{\displaystyle \frac{\sign(u - v) + 1}{2}} \mul (u - v) + v\;.
\end{ldispl}
We will write:
\begin{ldispl}
p > q    \;\;\mathrm{for}\;\;
{\displaystyle \frac{1 - \sign(p - q)}{1 - \sign(p - q)}} = 0\;, 
\qquad
p \leq q \;\;\mathrm{for}\;\;
{\displaystyle \frac{1 - \sign(p - q)}{1 - \sign(p - q)}} = 1\;.
\end{ldispl}

\section{Core Tuplix Calculus and Encapsulation}
\label{sect-CTC}

The timed tuplix calculus presented in this paper extends \CTC\
(Core Tuplix Calculus).
\CTC\ has been introduced in~\cite{BPZ07a} as the core of \TC\
(Tuplix Calculus).
In this section, we give a brief summary of \CTC\ and its extension
with encapsulation operators.
These operators have been introduced in~\cite{BPZ07a} as well.
The operators of the timed tuplix calculus that will be introduced in
Section~\ref{sect-TTC} include generalizations of the encapsulation
operators.

It is assumed that a fixed but arbitrary set $A$ of
\emph{transfer actions} has been given.
It is also assumed that a fixed but arbitrary signed non-trivial
cancellation meadow $\cD$ has been given.

\CTC\ has two sort: the sort $\Tuplix$ of \emph{tuplices} and the sort
$\Quantity$ of \emph{quantities}.
To build terms of sort $\Tuplix$, it has the following constants and
operators:
\begin{itemize}
\item
the \emph{empty tuplix} constant $\const{\emptx}{\Tuplix}$;
\item
the \emph{blocking tuplix} constant $\const{\nullx}{\Tuplix}$;
\item
for each $a \in A$, the unary \emph{transfer action} operator
$\funct{a}{\Quantity}{\Tuplix}$;
\item
the unary \emph{zero test} operator
$\funct{\gamma}{\Quantity}{\Tuplix}$;
\item
the binary \emph{conjunctive composition} operator
$\funct{\conjc}{\Tuplix \x \Tuplix}{\Tuplix}$.
\end{itemize}
To build terms of sort $\Quantity$, \CTC\ has the constants and
operators from the signature of meadows.

We assume that there are infinitely many variables of sort $\Tuplix$,
including $x$, $y$ and $z$, and infinitely many variables of sort
$\Quantity$, including $u$, $v$ and $w$.
Terms are build as usual for a many-sorted signature
(see e.g.~\cite{ST99a,Wir90a}).
We use infix notation for the binary operator $\conjc$.

A term of sort $\Tuplix$ is \emph{tuplix-closed} if it does not contain
variables of sort $\Tuplix$. \linebreak[2]
A term of sort $\Tuplix$ is \emph{closed} if it does not contain
variables of any sort.

We look at \CTC\ as a calculus that is concerned with transfers of
quantities of something.
Let $t$ and $t'$ be closed terms of sort $\Tuplix$, and
let $q$ be a closed term of sort $\Quantity$.
Intuitively, the constants and operators introduced above
can be explained as follows:
\begin{itemize}
\item
$\emptx$ is a tuplix with no effect;
\item
$\nullx$ blocks any joint effect of tuplices;
\item
the effect of $a(q)$ is performing action $a$ and transferring quantity
$q$ on performing that action;
\item
$\ztest{q}$ is a tuplix with no effect if $q$ equals $0$ and blocks any 
joint effect otherwise;
\item
the effect of $t \conjc t'$ is the joint effect of $t$ and $t'$.
\end{itemize}
In~\cite{BPZ07a}, these constants and operators are explained in a
different way.
We consider that way of explanation less appropriate for the timed
extension of \CTC\ that will be presented in Section~\ref{sect-TTC}.

We use the following convention: a transfer of a positive quantity is
taken as an outgoing transfer and a transfer of a negative quantity is
taken as an incoming transfer.

Notice that \CTC\ can be looked upon as a special purpose process
algebra in which processes are considered at a level of detail where not
even the order in which actions are performed matter.
This makes \CTC\ suitable for formalizing budgets: budgets are in fact
descriptions of financial behaviour at the level of detail where only
the actions to be performed and the quantities transferred on performing
those actions matter.

The axioms of \CTC\ are given in Table~\ref{axioms-CTC}.%
\begin{table}[!t]
\caption{Axioms of \CTC}
\label{axioms-CTC}
\begin{eqntbl}
\begin{axcol}
x \conjc y = y \conjc x                                 & \axiom{T1} \\
(x \conjc y) \conjc z = x \conjc (y \conjc z)           & \axiom{T2} \\
x \conjc \emptx = x                                     & \axiom{T3} \\
x \conjc \nullx  = \nullx                               & \axiom{T4} \\
a(u) \conjc a(v) = a(u + v)                             & \axiom{T5}
\end{axcol}
\qquad
\begin{axcol}
\ztest{u} = \ztest{u / u}                               & \axiom{T6} \\
\ztest{0} = \emptx                                      & \axiom{T7} \\
\ztest{1} = \nullx                                      & \axiom{T8} \\
\ztest{u} \conjc \ztest{v} = \ztest{u / u + v / v}      & \axiom{T9} \\
\ztest{u - v} \conjc a(u) = \ztest{u - v} \conjc a(v)   & \axiom{T10}
\end{axcol}
\end{eqntbl}
\end{table}
The following proof rule is adopted to lift the valid equations between
terms of sort $\Quantity$ to \CTC:
\begin{ldispl}
\mbox{for all terms $p$ and $q$ of sort $\Quantity$,} \quad
\cD \models p = q \quad \mathrm{implies} \quad \ztest{p} = \ztest{q}\;.
\end{ldispl}
We will refer to this proof rule by DE.

To prove a statement for all \CTC\ terms of sort $\Tuplix$, it is is
sufficient to prove it for all \CTC\ canonical terms.
A \emph{CTC canonical term} is a \CTC\ term of sort $\Tuplix$ of the
form
\begin{ldispl}
\ztest{p_0} \conjc a_1(p_1) \conjc \ldots \conjc a_k(p_k) \conjc
x_1 \conjc \ldots \conjc x_l\;,
\end{ldispl}
where $k,l \geq 0$ and $a_1,\ldots,a_k$ are distinct transfer actions.
\begin{lemma}
\label{lemma-CTC-can-term}
For all \CTC\ terms $t$ of sort $\Tuplix$, there exists a \CTC\
canonical term $t'$ such that $t = t'$ is derivable from the axioms of
\CTC.
\end{lemma}
\begin{proof}
This proposition is a reformulation of Lemma~1 from~\cite{BPZ07a}.
\qed
\end{proof}

Like in~\cite{BPZ07a}, we can add the following operators to the
operators of \CTC\ to build terms of sort $\Tuplix$:
\begin{itemize}
\item
for each $H \subseteq A$, the unary \emph{encapsulation} operator
$\funct{\encap{H}}{\Tuplix}{\Tuplix}$.
\end{itemize}

Let $t$ be a closed term of sort $\Tuplix$.
Intuitively, the encapsulation operators can be explained as follows:
\begin{itemize}
\item
if, for each $a \in H$, the sum of all quantities transferred by $t$ on
performing $a$ equals $0$, then $\encap{H}(t)$ differs from $t$ in that,
for each $a \in H$, the effect of all transfer actions of the form
$a(p)$ occurring in $t$ is eliminated; otherwise, $\encap{H}(t)$ has the
same effect as $\nullx$.
\end{itemize}
The name encapsulation was introduced earlier in the setting of the
process algebra \ACP\ for similar operations in~\cite{BK84b}.

The axioms for encapsulation are given in
Table~\ref{axioms-encap-pabstr}.%
\begin{table}[!t]
\caption{Axioms for encapsulation}
\label{axioms-encap-pabstr}
\begin{eqntbl}
\begin{axcol}
\encap{H}(\emptx) = \emptx                               & \axiom{E1} \\
\encap{H}(\nullx) = \nullx                               & \axiom{E2} \\
\encap{H}(\ztest{u}) = \ztest{u}                         & \axiom{E3}
\end{axcol}
\quad
\begin{axcol}
\encap{H}(a(u)) = a(u)      \hfill \mif a \notin H       & \axiom{E4} \\
\encap{H}(a(u)) = \ztest{u} \hfill \mif a \in H          & \axiom{E5} \\
\encap{H}(x \conjc \encap{H}(y)) = \encap{H}(x) \conjc \encap{H}(y)
                                                         & \axiom{E6} \\
\encap{H \union H'}(x) = \encap{H}(\encap{H'}(x))        & \axiom{E7}
\end{axcol}
\end{eqntbl}
\end{table}

\section{Timed Tuplix Calculus}
\label{sect-TTC}

In this section, we extend \CTC\ to \TTC\ (Timed Tuplix Calculus).
In the informal explanation of the constants and operators of \CTC\
given in Section~\ref{sect-CTC}, we could disregard what it is of which
quantities are transferred.
Clearly, if \CTC\ is used to formalize budgets, quantities of money are
transferred.
It happens to be far from obvious to give informal explanations of two
of the additional operators of \TTC\ that are not couched in terms of
quantities of money, usually called amounts of money.
Therefore, we change over in this section to explanations couched in
terms of amounts of money.
This should not be taken as a suggestion that more abstract explanations
are impossible.
In Section~\ref{sect-fpb}, tuplices are viewed as representations of
financial behaviours.
The change-over made in this section agrees with this viewpoint.

Like \CTC, \TTC\ has two sort: the sort $\Tuplix$ of tuplices and the
sort $\Quantity$ of quantities.
To build terms of sort $\Tuplix$, it has the constants and operators of
\CTC\ to build terms of sort $\Tuplix$, and in addition the following
operators:
\begin{itemize}
\item
the unary \emph{delay} operator $\funct{\delay}{\Tuplix}{\Tuplix}$;
\item
for each $I \subseteq A$, the unary \emph{pre-abstraction} operator
$\funct{\pabstr{I}}{\Tuplix}{\Tuplix}$;
\item
for each $H \subseteq A$, the binary
\emph{interest counting encapsulation} operator
$\funct{\iencap{H}{{}}}{\Quantity \x \Tuplix}{\Tuplix}$.
\end{itemize}
To build terms of sort $\Quantity$, it has the constants and operators
from the signature of meadows, and in addition the following operator:
\begin{itemize}
\item
the binary \emph{implicit capital} operator
$\funct{\icap{{}}}{\Quantity \x \Tuplix}{\Quantity}$.
\end{itemize}

We write $\iencap{H}{p}(t)$ and $\icap{p}(t)$, where $p$ is a term of
sort $\Quantity$ and $t$ is a term of sort~$\Tuplix$, for
$\iencap{H}{{}}(p,t)$ and $\icap{{}}(p,t)$, respectively.
We also use the notation $\delay^n(t)$.
For each term $t$ of sort $\Tuplix$, the term $\delay^n(t)$ is defined
by induction on $n$ as follows:
$\delay^0(t) = t$ and $\delay^{n+1}(t) = \delay(\delay^n(t))$.

In \TTC, it is assumed that $\iact \in A$.
A special role is assigned to $\iact$: transfer actions of the form
$a(p)$ are renamed to $\iact(p)$ on pre-abstraction in order to abstract
from their identity, but not from their presence.

We look at \TTC\ as a calculus that is concerned with transfers of 
amounts of money on time.
Let $t$ be a closed term of sort $\Tuplix$ and
let $p$ be a closed term of sort $\Quantity$.
Intuitively, the additional operators introduced above can be explained
as follows:
\begin{itemize}
\item
$\delay(t)$ differs from $t$ in that the effect of each transfer action
occurring in $t$ is delayed one time slice;
\item
$\pabstr{I}(t)$ differs from $t$ in that, for each $a \in I$, the effect
of each transfer action of the form $a(p)$ occurring in $t$ is replaced
by the effect of $\iact(p)$;
\item
$\iencap{H}{p}(t)$ differs from $\encap{H}(t)$ in that, for each
$a \in H$, a cumulative interest at the rate of $p$ per time slice is
taken into account on the summation of all amounts of money transferred
by $t$ on performing $a$;
\item
$\icap{p}(t)$ is the least amount of money that must be at disposal 
initially to allow for each transfer action occurring in $t$ to be
performed if a cumulative interest at the rate of $p$ per time slice is 
taken into account.
\end{itemize}
The delay operator introduced here is comparable to the relative
discrete time unit delay operator and the absolute discrete time unit
delay operator introduced earlier in the setting of the process algebra
\ACP\ in~\cite{BB95a}.
The pre-abstraction operators introduced here are comparable to the
pre-abstraction operators introduced earlier in the setting of the
process algebra \ACP\ in~\cite{BB88}.
The interest counting encapsulation operators are generalizations of the
encapsulation operators introduced in Section~\ref{sect-CTC}:
$\encap{H}(t)$ can be taken as abbreviation of $\iencap{H}{0}(t)$.
The implicit capital operator introduced here is comparable to the
implicit computational capital operator introduced earlier in the
setting of the process algebra \ACP\ in~\cite{BM06e}.

The implicit capital of a non-blocking tuplix is an amount of money that 
is non-negative, and the implicit capital of a blocking tuplix is 
undefined.
In order to circumvent the use of algebras with partial operations, $-1$
is used to represent the undefinedness of the implicit capital of a
blocking tuplix.

Notice that \TTC\ can be looked upon as a special purpose timed process
algebra in which processes are considered at a level of detail where the
time slices in which actions are performed matter, but not their order
within the time slices.
This makes \TTC\ suitable for analyzing financial products: financial
products involve transfers of amounts of money where the day, week or 
month in which actions are performed and the amounts of money that are 
transferred in doing so are relevant, but not their order within the 
periods concerned.

The axioms of \TTC\ are the axioms of \CTC\ and the additional axioms
given in Tables~\ref{axioms-delay-pabstr-iencap}
and~\ref{axioms-implicit-capital}.%
\begin{table}[!t]
\caption{Axioms for delay, pre-abstraction and interest counting
 encapsulation}
\label{axioms-delay-pabstr-iencap}
\begin{eqntbl}
\begin{axcol}
\delay(\emptx) = \emptx                                & \axiom{D1} \\
\delay(\nullx) = \nullx                                & \axiom{D2} \\
\delay(\ztest{u}) = \ztest{u}                          & \axiom{D3} \\
\delay(x \conjc y) = \delay(x) \conjc \delay(y)        & \axiom{D4} \\
{} \\
\pabstr{I}(\emptx) = \emptx                            & \axiom{PA1} \\
\pabstr{I}(\nullx) = \nullx                            & \axiom{PA2} \\
\pabstr{I}(\ztest{u}) = \ztest{u}                      & \axiom{PA3} \\
\pabstr{I}(a(u)) = a(u)        \hfill \mif a \notin I  & \axiom{PA4} \\
\pabstr{I}(a(u)) = \iact(u)    \hfill \mif a \in I     & \axiom{PA5} \\
\pabstr{I}(x \conjc y) = \pabstr{I}(x) \conjc \pabstr{I}(y)
                                                       & \axiom{PA6} \\
\pabstr{I}(\delay(x)) = \delay(\pabstr{I}(x))          & \axiom{PA7} \\
\pabstr{I \union I'}(x) = \pabstr{I}(\pabstr{I'}(x))   & \axiom{PA8}
\end{axcol}
\quad
\begin{axcol}
{} \\
{} \\
\ztest{1 - \frac{1+u}{1+u}} \conjc
\iencap{\set{a}}{u}(a(v) \conjc x) =
{} \\ \quad
\ztest{1 - \frac{1+u}{1+u}} \conjc
\iencap{\set{a}}{u}(\delay(a((1+u) \mul v)) \conjc x)  & \axiom{RICA} \\
{} \\
\iencap{H}{u}(\emptx) = \emptx                         & \axiom{ICE1} \\
\iencap{H}{u}(\nullx) = \nullx                         & \axiom{ICE2} \\
\iencap{H}{u}(\ztest{v}) = \ztest{v}                   & \axiom{ICE3} \\
\iencap{H}{u}(a(v)) = a(v)      \hfill \mif a \notin H & \axiom{ICE4} \\
\iencap{H}{u}(a(v)) = \ztest{v} \hfill \mif a \in H    & \axiom{ICE5} \\
\iencap{H}{u}(x \conjc \iencap{H}{u}(y)) =
\iencap{H}{u}(x) \conjc \iencap{H}{u}(y)               & \axiom{ICE6} \\
\iencap{H}{u}(\delay(x)) = \delay(\iencap{H}{u}(x))    & \axiom{ICE7} \\
\iencap{H \union H'}{u}(x) = \iencap{H}{u}(\iencap{H'}{u}(x))
                                                       & \axiom{ICE8}
\end{axcol}
\end{eqntbl}
\end{table}
\begin{table}[!t]
\caption{Axioms for implicit capital}
\label{axioms-implicit-capital}
\begin{eqntbl}
\begin{axcol}
\icap{u}(x) = \icap{u}(\pabstr{A}(x))                  & \axiom{IC1} \\
\icap{u}(\emptx) = 0                                   & \axiom{IC2} \\
\icap{u}(\nullx) = -1                                  & \axiom{IC3} \\
\icap{u}(\iact(v)) = \max(v,0)                         & \axiom{IC4} \\
\frac{1 + \icap{u}(x)}{1 + \icap{u}(x)} \mul \icap{u}(\delay(x)) =
\frac{1 + \icap{u}(x)}{1 + \icap{u}(x)} \mul
\max(\frac{1}{1 + u} \mul \icap{u}(x),0)               & \axiom{IC5} \\
\frac{1 + \icap{u}(x)}{1 + \icap{u}(x)} \mul
\icap{u}(\iact(v) \conjc \delay(x)) =
\frac{1 + \icap{u}(x)}{1 + \icap{u}(x)} \mul
\max(v + \frac{1}{1 + u} \mul \icap{u}(x),0)           & \axiom{IC6}
\end{axcol}
\end{eqntbl}
\end{table}
Like in \CTC, the proof rule DE is adopted to lift the valid equations
between terms of sort $\Quantity$ to \TTC.

Axiom RICA (Realistic Interest Calculation Axiom) is equivalent to
\begin{ldispl}
u \neq -1 \Implies
\iencap{\set{a}}{u}(a(v) \conjc x) =
\iencap{\set{a}}{u}(\delay(a((1+u) \mul v)) \conjc x)\;.
\end{ldispl}
This formula can be paraphrased as follows: when encapsulating $a$, 
reckoning with an interest rate $u$ different from $-1$, an undelayed 
transfer of an amount $v$ is equivalent to a transfer of an amount
$(u + 1) \mul v$ in the next time slice.
The exclusion of $u = -1$ prevents that the equation $x = \nullx$ can be
derived.
Axioms IC5 and IC6 are equivalent to
\begin{ldispl}
\icap{u}(x) \neq -1 \Implies
\icap{u}(\delay(x)) =
\max(\frac{1}{1 + u} \mul \icap{u}(x),0)\;,
\eqnsep
\icap{u}(x) \neq -1 \Implies
\icap{u}(\iact(v) \conjc \delay(x)) =
\max(v + \frac{1}{1 + u} \mul \icap{u}(x),0)\;.
\end{ldispl}
These formulas express that, reckoning with an interest rate $u$, the 
total contribution of all transfers in the next time slice to the 
implicit capital equals $\frac{1}{1 + u}$ times what their total 
contribution would be in the current time slice.
The exclusion of $\icap{u}(x) = -1$ is needed because $-1$ is used to
represent undefinedness.

\begin{example}
\label{example-iencap}
Let $p$ be a closed term of sort $\Quantity$ such that
$\cD \models \frac{1+p}{1+p} = 1$.
The following is a derivation from the axioms of \TTC\ and the proof
rule DE:
\begin{ldispl}
\begin{aeqns}
   &   &
\iencap{\set{a}}{p}(a(u) \conjc \delay(a(5)) \conjc \delay^2(b(u - 7)))
\\ & = &
\iencap{\set{a}}{p}(a(u) \conjc a(\frac{5}{1+p}) \conjc
                    \delay^2(\iencap{\set{a}}{p}(b(u - 7))))
\\ & = &
\iencap{\set{a}}{p}(a(u + \frac{5}{1+p}) \conjc
                    \iencap{\set{a}}{p}(\delay^2(b(u - 7))))
\\ & = &
\iencap{\set{a}}{p}(a(u + \frac{5}{1+p})) \conjc
\iencap{\set{a}}{p}(\delay^2(b(u - 7)))
\\ & = &
\iencap{\set{a}}{p}(a(u + \frac{5}{1+p})) \conjc
\delay^2(\iencap{\set{a}}{p}(b(u - 7)))
\\ & = &
\ztest{u + \frac{5}{1+p}} \conjc \delay^2(b(u - 7))\;.
\end{aeqns}
\end{ldispl}
Because $\cD \models \frac{-5}{1+p} + \frac{5}{1+p} = 0$, it follows
immediately that
\begin{ldispl}
\iencap{\set{a}}{p}(a(\frac{-5}{1+p}) \conjc \delay(a(5)) \conjc
                    \delay^2(b(\frac{-5}{1+p} - 7))) =
\delay^2(b(\frac{-5}{1+p} - 7))\;.
\end{ldispl}
Moreover, it follows immediately that
\begin{ldispl}
\iencap{\set{a}}{p}(a(q) \conjc \delay(a(5)) \conjc
                    \delay^2(b(q - 7))) = \nullx
\end{ldispl}
for all closed terms $q$ of sort
$\Quantity$ such that not $\cD \models q + \frac{5}{1+p} = 0$.
\end{example}

\begin{example}
\label{example-icap}
Let $p$ and $q$ be closed terms of sort $\Quantity$.
The following is a derivation from the axioms of \TTC\ and the proof
rule DE:
\pagebreak[2]
\begin{ldispl}
\begin{aeqns}
   &   &
\icap{p}
 (a(7) \conjc \delay(a'(-8)) \conjc
  b(-5) \conjc \delay^2(b'((1+q)^2 \mul 5)))
\\ & = &
\icap{p}
 (\iact(7) \conjc \delay(\iact(-8)) \conjc
  \iact(-5) \conjc \delay^2(\iact((1+q)^2 \mul 5)))
\\ & = &
\icap{p}
 (\iact(2) \conjc
  \delay(\iact(-8) \conjc \delay(\iact((1+q)^2 \mul 5))))
\\ & = &
\max(2 + \frac{1}{1+p} \mul
     \icap{p}(\iact(-8) \conjc \delay(\iact((1+q)^2 \mul 5))),
     0)
\\ & = &
\max(2 + \frac{1}{1+p} \mul
         \max(-8 + \frac{1}{1+p} \mul
                   \icap{p}(\iact((1+q)^2 \mul 5)),0),0)
\\ & = &
\max(2 +
     \frac{1}{1+p} \mul
     \max(-8 + \frac{1}{1+p} \mul (1+q)^2 \mul 5,0),0)\;.
\end{aeqns}
\end{ldispl}
It follows immediately that
\begin{ldispl}
\icap{p}
 (a(7) \conjc \delay(a'(-8)) \conjc
  b(-5) \conjc \delay^2(b'((1+q)^2 \mul 5))) = 2
\end{ldispl}
for all closed terms $p$ and $q$ of sort $\Quantity$ such that
$\cD \models \frac{1}{1+p} \mul (1+q)^2 \leq \frac{8}{5}$.
There are many such $p$ and $q$, for example, $p$ and $q$ such that
$\cD \models p = \frac{1}{100}$ and $\cD \models q = \frac{10}{100}$,
but also $p$ and $q$ such that $\cD \models p = \frac{25}{100}$ and
$\cD \models q = \frac{40}{100}$. \linebreak[2]
We will return to this example in Section~\ref{sect-fpb}.
\end{example}

To prove a statement for all tuplix-closed \TTC\ terms of sort
$\Tuplix$, it is sufficient to prove it for all tuplix-closed \TTC\
canonical terms.
The set of \emph{TTC canonical terms} is inductively defined by the
following rules:
\begin{itemize}
\item
if $t$ is a \CTC\ canonical term, then $t$ is a \TTC\ canonical term;
\item
if $t$ is a \CTC\ canonical term and $t'$ is a \TTC\ canonical term,
then $t \conjc \delay(t')$ is a \TTC\ canonical term.
\end{itemize}
\begin{lemma}
\label{lemma-TTC-can-term}
For all tuplix-closed \TTC\ terms $t$ of sort $\Tuplix$, there exists a
tuplix-closed \TTC\ canonical term $t'$ such that $t = t'$ is derivable
from the axioms of \TTC.
\end{lemma}
\begin{proof}
The proof is straightforward by induction on the structure of $t$, and
in the cases $t \equiv \pabstr{I}(s)$ and $t \equiv \iencap{H}{p}(s)$
(where we can restrict ourselves to tuplix-closed \TTC\ canonical terms
$s$) by induction on the structure of $s$.
The following easy to prove fact is used in the proof for the case
$t \equiv \iencap{H}{p}(s)$:
for all \TTC\ terms $t_1$ of sort $\Tuplix$ and all tuplix-closed \TTC\
terms $t_2$ of sort $\Tuplix$ in which no element of $H$ occurs,
$\iencap{H}{u}(t_1 \conjc t_2) = \iencap{H}{u}(t_1) \conjc t_2$ is
derivable from the axioms of \TTC\ \linebreak[2]
(cf.\ Lemma~5 in~\cite{BPZ07a}).
\qed
\end{proof}
The following is a useful corollary of Lemma~\ref{lemma-TTC-can-term}.
\begin{corollary}
\label{corol-TTC-can-term}
For all tuplix-closed \TTC\ terms $t$ of sort $\Tuplix$, there exists a
tuplix-closed \TTC\ term $t'$ of the form
$\delay^0(t_0) \conjc \ldots \conjc \delay^n(t_n)$,
where $n \geq 0$ and $t_0,\ldots,t_n$ are tuplix-closed \CTC\ canonical 
terms, such that $t = t'$ is derivable from the axioms of \TTC.
\end{corollary}

\section{Financial Product Behaviours}
\label{sect-fpb}

In this section, we formalize the cumulative interest compliant 
conservation requirement proposed by Wesseling and van den Bergh, use
this formalization to introduce the notion of a financial product 
behaviour, and present some properties of financial product behaviours.
We use \TTC\ for this, viewing tuplices as representations of financial 
behaviours.

Here, the signed cancellation meadow $\cD$, which is a parameter of 
\TTC, is confined to the signed meadow of real numbers.
The signed meadow of rational numbers would not serve our purpose as 
will be explained hereafter.

In~\cite{WB00a}, Wesseling and van den Bergh claim that interest
calculations relating to financial products should always be based on
cumulative interests.
By strictly adhering to the use of cumulative interests, the design of
financial products is made symmetric between client and provider and an
implicit bias towards either party can be avoided.
This is the point of departure of their `realistic interest calculation
approach' and the origin of axiom RICA of TTC.
Applying this approach involves a strict separation between transfers
related to a financial product proper and transfers related to its costs
of delivery.
Transfers related to the financial product proper include transfers due
to interests.
Transfers related to the costs of delivery may include clear profit,
general running cost, cost of insurance against non-payment, costs of
marketing and communication, etc.

Having made this separation, Wesseling and van den Bergh formulate a
cumulative interest compliant conservation requirement for financial 
products: the sum of all transfers relating to the product, transposed 
to some point of time (the focal date) by means of cumulative interest 
at the effective interest rate of the product, is zero.
In~\cite{WB00a}, this requirement is presented in the form of an 
equation whose left-hand side and right-hand side are informally 
described.
The equation concerned has two unknowns, to wit a financial behaviour 
and an interest rate.
If a financial behaviour and an interest rate make up a solution of the
equation, then the interest rate is taken for the effective interest 
rate of a financial product that involves the financial behaviour.

The cumulative interest compliant conservation requirement for financial
products is formalized in TTC by the equation
\begin{ldispl}
\iencap{\set{\iact}}{u}(\pabstr{A}(x)) = \emptx\;.
\end{ldispl}
This equation is called the \emph{Wesseling and van den Bergh equation}
or shortly the \emph{W-vdB equation}.
In the following definition, we make use of the W-vdB equation.
The definition is inspired by the perspective mentioned at the end of 
the last paragraph.
Let $t$ be a closed term of sort $\Tuplix$.
Then $t$ represents a \emph{financial product behaviour} if 
\begin{ldispl}
\Exists{u}
 {(\Forall{v}
    {(v > -1 \Implies 
      (\iencap{\set{\iact}}{v}(\pabstr{A}(t)) = \emptx \Iff
       u = v)})}\;.\footnotemark
\end{ldispl}
\footnotetext
{It follows from the decidability of the first-order theory of real 
 numbers with addition, multiplication and order (see~\cite{Tar51a})
 that it is decidable whether a closed term of sort $\Tuplix$
 represents a financial product behaviour.}
We see that the interest rate $v$ for which the equation
$\iencap{\set{\iact}}{v}(\pabstr{A}(t)) = \emptx$
holds must meet the condition that $v > -1$ and the condition that $v$ 
is the unique interest rate meeting the first condition for which the 
equation holds.
These conditions are healthiness conditions: if they are not met, we 
have to do with an implausible financial product behaviour.
Instead of a uniqueness condition on $v$, we could have a condition on
$t$ based on Descartes' rule of signs or one of its relatives 
(see e.g.~\cite{Mes82a}).
Thus, we would have replaced the uniqueness condition on $v$ by a 
condition that is sufficient but not necessary for the uniqueness of 
$v$.
That is, we would have a less general definition. 

Each closed term of sort $\Tuplix$ represents a financial behaviour, but 
not each closed term of sort $\Tuplix$ represents a financial product 
behaviour.
A financial product behaviour can be seen as a financial behaviour for 
which a financial product can be devised that involves the behaviour.
However, a financial product behaviour may also have one or more origins 
different from a financial product.
For example, viewed apart, the financial behaviour that is part of a 
trading behaviour is often a financial product behaviour as well.

The definition of a financial product behaviour given above agrees with
the viewpoint that a financial product entails an agreement under which 
a party gives one or more fixed amounts of money to another party, each 
of them at a fixed date, with the understanding that the former party 
will get back one or more fixed amounts of money, each of them at a 
fixed date (freely cited\linebreak[2] from~\cite{Fei07a}).

Consider a loan of \euro1,000 for which the borrower has to pay back 
\euro2,000 after two years.
The financial behaviour involved in this loan is a financial product 
behaviour according to the definition given above only if the equation
\smash{$-1000 + \frac{2000}{(1 + v)^2} = 0$} has a unique solution 
greater than $-1$.
This equation has a unique solution greater than~$-1$ in the 
signed meadow of real numbers, to wit $\surd2 - 1$, but no solution in 
the signed meadow of rational numbers.
This example shows that there are genuine financial products that 
involve financial behaviours which would not be financial product
behaviours according to the definition given above if interest rates
would be restricted to rational numbers.
Another matter is that in reality financial institution cannot help
but approximate interest rates like $\surd2 - 1$ with some finite 
accuracy.

Let $p$ be a closed term of sort $\Quantity$ and $t$ be a closed term of
sort $\Tuplix$ such that 
$\iencap{\set{\iact}}{p}(\pabstr{A}(t)) = \emptx$.
Then $t$ represents a financial product behaviour and $p$ represents the 
effective interest rate of the underlying financial product.
If that financial product is a financial product of credit type, then 
$\icap{p}(t) = 0$.
However, if that financial product is a financial product of savings 
type, then $\icap{p}(t) > 0$.

Let $p$ and $q$ be closed terms of sort $\Quantity$ and $t$ and $t'$
be closed terms of sort $\Tuplix$ such that 
\smash{$\iencap{\set{\iact}}{q}(\pabstr{A}(t)) = \emptx$}, 
$\icap{q}(t) = 0$, and $\icap{p}(t') > 0$.
Then we say that the financial behaviour $t'$ profits from using the 
financial product underlying $t$ taking the interest rate $p$ into 
account if $\icap{p}(t \conjc t') < \icap{p}(t')$.
In any case, \linebreak[2] we have 
$\icap{p}(t \conjc t') \leq \icap{p}(t) + \icap{p}(t')$.
The important observation is that we may have 
$\icap{p}(t \conjc t') < \icap{p}(t')$.
\begin{proposition}
\label{prop-profit-of-loan}
There exist closed terms $p$ and $q$ of sort $\Quantity$ and closed 
terms $t$ and $t'$ of sort $\Tuplix$ with
$\iencap{\set{\iact}}{q}(\pabstr{A}(t)) = \emptx$, $\icap{q}(t) = 0$, 
and $\icap{p}(t') > 0$ such that $\icap{p}(t \conjc t') < \icap{p}(t')$.
\end{proposition}
\begin{proof}
Take the case where $p$ and $q$ are such that
$\cD \models \frac{1}{1+p} \mul (1+q)^2 \leq \frac{8}{5}$,
$t \equiv b(-5) \conjc \delay^2(b'((1+q)^2 \mul 5))$, and
$t' \equiv a(7) \conjc \delay(a'(-8))$.
We can easily derive that 
$\iencap{\set{\iact}}{q}(\pabstr{A}(t)) = \emptx$, 
$\icap{q}(t) = 0$, and $\icap{p}(t') = 7$.
Moreover, in Example~\ref{example-icap}, we have already derived that
$\icap{p}(t \conjc t') = 2$.
Hence, $\icap{p}(t \conjc t') < \icap{p}(t')$.
\qed
\end{proof}
Proposition~\ref{prop-profit-of-loan} can be read as follows: there
exists an interest rate, a financial product of credit type, and a 
financial behaviour that profits from that financial product if that 
interest rate is taken into account.

\begin{proposition}
\label{prop-fpb-min}
Let $t$ and $t'$ be closed terms of sort $\Tuplix$ such that $t'$ is $t$
with each subterm of the form $a(p)$ replaced by $a(-p)$, and let $q$ be
a closed term of sort $\Quantity$ such that $q \neq -1$.
Then $\iencap{\set{\iact}}{q}(\pabstr{A}(t)) = \emptx$ implies
$\iencap{\set{\iact}}{q}(\pabstr{A}(t')) = \emptx$.
\end{proposition}
\begin{proof}
Assume that $\iencap{\set{\iact}}{q}(\pabstr{A}(t)) = \emptx$.
Then $t \neq \nullx$.
From this and Corollary~\ref{corol-TTC-can-term}, it follows that 
$\pabstr{A}(t)$ is of the form 
$\delay^0(t_0) \conjc \ldots \conjc \delay^n(t_n)$, where 
$t_0,\ldots,t_n$ are of the form $\iact(p)$ or $\emptx$.
For each $i \in \set{0,\ldots,n}$,
let $p_i$ be such that $\iact(p_i) \equiv t_i$ if 
$t_i \not\equiv \emptx$ and $p_i \equiv 0$ if $t_i \equiv \emptx$.
Then 
$\iencap{\set{\iact}}{q}(\pabstr{A}(t)) = 
 \ztest{\sum_{i =0}^n \frac{1}{(1+q)^i} \mul p_i}$ and
$\iencap{\set{\iact}}{q}(\pabstr{A}(t')) = 
 \ztest{\sum_{i =0}^n \frac{1}{(1+q)^i} \mul -p_i}$.
Because $\iencap{\set{\iact}}{q}(\pabstr{A}(t)) = \emptx$, we know that
$\sum_{i =0}^n \frac{1}{(1+q)^i} \mul p_i = 0$.
From this and the fact that  
$\sum_{i =0}^n \frac{1}{(1+q)^i} \mul -p_i = 
 -\sum_{i =0}^n \frac{1}{(1+q)^i} \mul p_i$,
it follows that 
$\sum_{i =0}^n \frac{1}{(1+q)^i} \mul -p_i = 0$.
Hence, $\iencap{\set{\iact}}{q}(\pabstr{A}(t')) = \emptx$.
\qed
\end{proof}
Proposition~\ref{prop-fpb-min} can be read as follows: if we change the 
incoming transfers of a financial product into outgoing transfers and
its outgoing transfers into incoming transfers, then the result is a 
financial product behaviour as well; and the effective interest rates of
the underlying financial products are the same.

Let $t$ and $t'$ be closed terms of sort $\Tuplix$.
Then $t$ is a \emph{time inverse} of $t'$ if, for some natural number 
$n$, there exist closed \CTC\ canonical terms $t_0,\ldots,t_n$ such that
$t = \delay^0(t_0) \conjc \ldots \conjc \delay^n(t_n)$ and
$t' = \delay^0(t_n) \conjc \ldots \conjc \delay^n(t_0)$.
If follows immedi\-ately from the definition that $t$ is a time inverse 
of $t'$ if and only if $t'$ is a time inverse of $t$.
By Corollary~\ref{corol-TTC-can-term}, each closed term of sort 
$\Tuplix$ has a time inverse.
This time inverse is unique up to derivable equality.

\begin{proposition}
\label{prop-fpb-inv}
Let $t$ and $t'$ be closed terms of sort $\Tuplix$ such that $t$ is a
time inverse of $t'$, and let $p$ and $q$ be closed terms of sort 
$\Quantity$ such that $p \neq -1$ and \smash{$q = \frac{-p}{1+p}$}.
Then $\iencap{\set{\iact}}{p}(\pabstr{A}(t)) = \emptx$ implies
$\iencap{\set{\iact}}{q}(\pabstr{A}(t')) = \emptx$.
\end{proposition}
\begin{proof}
Assume that $\iencap{\set{\iact}}{q}(\pabstr{A}(t)) = \emptx$.
Then $t \neq \nullx$.
From this and Corollary~\ref{corol-TTC-can-term}, it follows that 
$\pabstr{A}(t)$ is of the form 
$\delay^0(t_0) \conjc \ldots \conjc \delay^n(t_n)$, where 
$t_0,\ldots,t_n$ are of the form $\iact(p)$ or $\emptx$.
For each $i \in \set{0,\ldots,n}$,
let $p_i$ be such that $\iact(p_i) \equiv t_i$ if 
$t_i \not\equiv \emptx$ and $p_i \equiv 0$ if $t_i \equiv \emptx$.
Then 
$\iencap{\set{\iact}}{p}(\pabstr{A}(t)) = 
 \ztest{\sum_{i =0}^n \frac{1}{(1+p)^i} \mul p_i}$ and
$\iencap{\set{\iact}}{q}(\pabstr{A}(t')) = 
 \ztest{\sum_{i =0}^n \frac{1}{(1+q)^{n-i}} \mul p_i}$.
Because $\iencap{\set{\iact}}{p}(\pabstr{A}(t)) = \emptx$, we know that
$\sum_{i =0}^n \frac{1}{(1+p)^i} \mul p_i = 0$.
From this and the fact that  
$\sum_{i =0}^n \frac{1}{(1+q)^{n-i}} \mul p_i =
 (1 + p)^n \mul \sum_{i =0}^n \frac{1}{(1+p)^i} \mul p_i$,
it follows that 
$\sum_{i =0}^n \frac{1}{(1+q)^{n-i}} \mul p_i = 0$.
Hence, $\iencap{\set{\iact}}{q}(\pabstr{A}(t')) = \emptx$.
\qed
\end{proof}
Proposition~\ref{prop-fpb-inv} can be read as follows: if we reverse the 
order of time in which the transfers of a financial product behaviour 
take place, then the result is a financial product behaviour as well; 
and if the effective interest rate of the former financial products is 
$p$ then the effective interest rate of the latter financial products 
is~$\frac{-p}{1+p}$.

\section{Standard Model of \TTC}
\label{sect-model}

In this section, we construct the standard model of \TTC.
The standard model of \CTC\ presented in~\cite{BPZ07a} lies at the root
of this model.
However, the use of partial functions is circumvented.

We write $\cD$ for the domain of the signed cancellation meadow $\cD$,
and we write $\op$, where $\op$ is a constant or operator from the
signature of signed cancellation meadows, for the interpretation of
$\op$ in $\cD$.
To prevent confusion with the constants from the signature of meadows, 
we write $\natzero$ and $\natone$ for the identity elements of addition 
and multiplication on natural numbers.

We define the set $\TE$ of \emph{tuplix elements}, the set $\UT$ of
\emph{untimed tuplices}, and the set $\TT$ of \emph{timed tuplices} as
follows:
\begin{ldispl}
\begin{aeqns}
\TE & = &
\Union{A' \subseteq \Attr} (\mapof{A'}{\cD})\;,
\eqnsep
\UT & = &
\set{U \subseteq \TE \where \card(U) \leq \natone}\;,
\eqnsep
\TT & = &
\set{\funct{T}{\Nat}{\UT} \where
     \Forall{i \in \Nat}{(\card(T(i)) = \natzero)} \Or
     \Forall{i \in \Nat}{(\card(T(i)) = \natone)}}\;.
\end{aeqns}
\end{ldispl}
In the definition of the standard model of \TTC, the auxiliary set
$\TT^-$ defined by
\begin{ldispl}
\TT^- =
 \set{T \in \TT \where \Forall{i \in \Nat}{(\card(T(i)) = \natone)}}
\end{ldispl}
is used as well.
We write $\elem(U)$, where $U \in \UT$, for the unique element
$f \in \TE$ such that $f \in U$ if $\card(U) = \natone$, and an
arbitrary $f \in \TE$ otherwise.

The \emph{standard model} of \TTC, written $\cM(\cD,A)$, is the
expansion of the signed cancellation meadow $\cD$ with
\begin{itemize}
\item
for the sort $\Tuplix$, the set $\TT$;
\item
for each additional constant $\const{\op_0}{\Tuplix}$ of \TTC, the
element $\opi_0 \in \TT$ defined in Table~\ref{interpretation-TTC};
\item
for each additional operator
$\funct{\op_n}{S_1 \x \ldots \x S_n}{S_{n+1}}$ of \TTC, the operation
$\funct{\opi_n}{D_1 \x \ldots D_n}{D_{n+1}}$, where $D_i = \TT$ if
$S_i \equiv \Tuplix$ and $D_i = \cD$ if $S_i \equiv \Quantity$, defined
in Table~\ref{interpretation-TTC}.%
\footnote
{We write $\emptymap$ for the empty function and $\maplet{e}{e'}$ for
 the function $f$ with $\dom(f) = \set{e}$ such that $f(e) = e'$.}
\end{itemize}
\begin{table}[!t]
\caption{Interpretation of constants and operators of \TTC}
\label{interpretation-TTC}
\begin{eqntbl}
\begin{caeqns}
\memptx(i) & = & \set{\emptymap}
\\
\mnullx(i) & = & \emptyset
\\
\mentry{a}{d}(i) & = &
\left\{\begin{col}
       \set{\maplet{a}{d}} \\ \set{\emptymap}
       \end{col}
\right.
       &
       \begin{col}
       \mif i = \natzero   \\ \mother
       \end{col}
\\
\mztest{d}(i) & = &
\left\{\begin{col}
       \set{\emptymap} \\ \emptyset
       \end{col}
\right.
       &
       \begin{col}
       \mif d = \mzero  \\ \mother
       \end{col}
\\
(T \mconjc T')(i) & = &
\set{f \aconjc f' \where f \in T(i) \And f' \in T'(i)}
\\
\mdelay(T)(i) & = &
\left\{\begin{col}
       T(i - \natone) \\
       \set{\emptymap} \\
       \emptyset
       \end{col}
\right.
       &
       \begin{col}
       \mif i > \natzero \And T(i) \neq \emptyset \\
       \mif i = \natzero \And T(i) \neq \emptyset \\
       \mother
       \end{col}
\\
\mpabstr{I}(T)(i) & = &
\set{\apabstr{I}(f) \where f \in T(i)}
\\
\miencap{H}{d}(T)(i) & = &
{\set{\aclear{H}(f) \where
      f \in T(i) \And \Forall{a \in H}{(\atotal{a}{d}(T) = 0)}}}
\eqnsep
\micap{d}(T) & = &
\left\{\begin{col}
       \aicap{d}(T) \\
       -1
       \end{col}
\right.
       &
       \begin{col}
       \mif \Exists{i \geq \natzero}{(T(i) \neq \emptyset)} \\
       \mother
       \end{col}
\end{caeqns}
\end{eqntbl}
\end{table}
In Table~\ref{interpretation-TTC}, the following auxiliary functions are
used:
\begin{itemize}
\item
the function $\funct{{\aconjc}}{\TE \x \TE}{\TE}$ defined by
\begin{itemize}
\item
$\dom(f \aconjc f') = \dom(f) \union \dom(f')$;
\item
for each $a \in \dom(f \aconjc f')$:
% !!
\begin{ldispl}
\hsp{-1.65}
\begin{caeqns}
(f \aconjc f')(a) & = &
\left\{\begin{col}
       f(a) \madd f'(a) \\
       f(a) \\
       f'(a)
       \end{col}
\right.
       &
       \begin{col}
       \mif a \in \dom(f) \inter \dom(f') \\
       \mif a \in \dom(f) \diff \dom(f') \\
       \mif a \in \dom(f') \diff \dom(f)\;;
       \end{col}
\end{caeqns}
\end{ldispl}
\end{itemize}
\pagebreak[2]
\item
for each $I \subseteq \Attr$,
the function $\funct{\apabstr{I}}{\TE}{\TE}$ defined by
\begin{itemize}
\item
$\dom(\apabstr{I}(f)) =
 (\dom(f) \diff I) \union
 \set{\iact \where \dom(f) \inter I \neq \emptyset}$;
\item
for each $a \in \dom(\apabstr{I}(f))$:
% !!
\begin{ldispl}
\hsp{-1.65}
\begin{caeqns}
\apabstr{I}(f)(a) & = &
\left\{\begin{col}
       f(a) \\
       \sum_{a' \in I} f(a')
       \end{col}
\right.
       &
       \begin{col}
       \mif a \neq \iact \\
       \mif a =    \iact\;;
       \end{col}
\end{caeqns}
\end{ldispl}
\end{itemize}
\item
for each $H \subseteq \Attr$,
the function $\funct{\aclear{H}}{\TE}{\TE}$ defined by
\begin{itemize}
\item
$\dom(\aclear{H}(f)) = \dom(f) \diff H$;
\item
for each $a \in \dom(\aclear{H}(f))$:
% !!
\begin{ldispl}
\hsp{-1.65}
\begin{caeqns}
\aclear{H}(f)(a) & = & f(a)\;;
\end{caeqns}
\end{ldispl}
% !!
\vspace*{.75ex}\par
\end{itemize}
\item
for each $a \in \Attr$,
the function $\funct{\atotal{a}{}}{\cD \x \TT}{\cD}$ defined by
\begin{ldispl}
\begin{caeqns}
\atotal{a}{d}(T) & = &
{\displaystyle \sum_{i \;\mathrm{s.t.}\; a \in \dom(\elem(T(i)))}}
 (1 \madd d)^i \mmul \elem(T(i))(a)\;;
\end{caeqns}
\end{ldispl}
\item
the function $\funct{\aicap{}}{\cD \x \TT^-}{\cD}$ recursively defined 
by
\begin{ldispl}
\begin{caeqns}
\aicap{u}(T) & = &
\left\{\begin{col}
       \max(q_0(T),\mzero) \\
       \max(q_0(T) \madd
            {\displaystyle \frac{1}{1 \madd u}} \mmul 
            \aicap{u}(\shift(T)),\mzero)
       \end{col}
\right.
       \begin{col}
       \mif \Forall{i > \natzero}{(T(i) =    \set{\emptymap})} \\
       \mif \Exists{i > \natzero}{(T(i) \neq \set{\emptymap})}\;,
       \end{col}
\end{caeqns}
\end{ldispl}
where:
\begin{itemize}
\item
$\funct{\shift}{\TT^-}{\TT^-}$ is defined by
$\shift(T)(i) = T(i + \natone)$ for all $i \in \Nat$;
\item
$\funct{q_0}{\TT^-}{\cD}$ is defined by
$q_0(T) = \sum_{a \in \dom(\elem(T(\natzero)))} \elem(T(\natzero))(a)$.
\end{itemize}
\end{itemize}

It is easy to establish the following soundness result: for all terms
$t$ and $t'$ of sort $\Tuplix$, $t = t'$ is derivable from the axioms of
\TTC\ and the proof rule DE only if $\cM(\cD,A) \models t = t'$.
We also have a completeness result.
\begin{theorem}
\label{theorem-completeness}
For all closed terms $t$ and $t'$ of sort $\Tuplix$,
$\cM(\cD,A) \models t = t'$ only if $t = t'$ is derivable from the
axioms of \TTC\ and the proof rule \textup{DE}.
\end{theorem}
\begin{proof}
By Lemma~\ref{lemma-TTC-can-term}, it is sufficient to show that,
for all closed \TTC\ canonical terms $t$ and $t'$,
$\cM(\cD,A) \models t = t'$ only if $t = t'$ is derivable from the
axioms of \TTC\ and the proof rule \textup{DE}.
This is easy to prove by induction on the structure of $t$ using
Theorem~1 from~\cite{BPZ07a}.
\qed
\end{proof}

\section{Concluding Remarks}
\label{sect-conclusions}

We have developed a timed extension of the core of tuplix calculus in 
which financial behaviours are considered at a level of detail where the 
time slices in which actions are performed matter, but not their order 
within the time slices.
This makes it suited for the description and analysis of financial
products: financial products exhibit financial behaviours where the day,
week or month in which actions are performed and the amounts of
money are transferred in doing so are relevant, but not their order
within the periods concerned.

We have formalized the cumulative interest compliant conservation
requirement for financial products proposed by Wesseling and van den 
Bergh by an equation in the timed tuplix calculus developed.
Thus, a formalization of the starting-point of the material on the
mathematics of finance presented in~\cite{WB00a} has been achieved.
Moreover, we have used this formalization to introduce the notion of a 
financial product behaviour, and have presented some properties of 
financial product behaviours.
The timed tuplix calculus appears to be a reasonable setting for further
work in this area.

In~\cite{BPZ07a}, the core of tuplix calculus is among other things 
extended with a binary alternative composition operator and a 
variable-binding generalized alternative composition operator for each 
variable of sort $\Quantity$.
The latter operators have proved to be convenient in modular budget 
design.
Extending timed tuplix calculus with these operators would allow for
non-deterministic financial behaviours to be described.
However, in the presence of non-deterministic financial behaviours it 
would be less easy to acquire an intuitive understanding of what the 
implicit capital of a financial behaviour tells us.
Moreover, comparison of the implicit capitals of different financial 
behaviours, like in Section~\ref{sect-fpb}, appears to make little sense
in the case of non-deterministic financial behaviours.

Like Wesseling and van den Bergh, we consider only financial products of
which the interest rate is not dependent on changes in the financial
market.
If the interest rate of a financial product is made dependent on changes
in the financial market, then the expressiveness of the timed tuplix
calculus is insufficient.
In this more dynamic case, a version of discrete time process 
algebra~\cite{BB95a} looks to be a reasonable setting for the 
formalization of an adapted version of the cumulative interest compliant 
conservation requirement.

We remark that we do not have to abandon discrete time if interest is 
continuously instead of discretely compounded because of the commonly 
known fact that an interest rate $p$ with continuous compounding is 
equivalent to an interest rate $\ln(1 + p)$ with discrete compounding.

We mention that the cumulative interest compliant conservation 
requirement for financial products has been formulated by Wesseling and 
van den Bergh under the influence of basic ideas on the mathematics of 
finance presented in~\cite{CCF90a}.

The work to which ours seems to be most closely related is the work on 
MLFi (Modeling Language for Finance)~\cite{Lex05a}.
MLFi is a language to describe financial products in a mathematically 
precise, compositional way.
A distinctive feature of MLFi is that the descriptions of financial 
products can be analyzed, manipulated, and translated in many ways.
Therefore, MLFi is considered to be the basis of an approach to the
application of various formal methods in matters concerning financial 
products.
TTC could find a place among these formal methods.

\bibliographystyle{splncs03}
\bibliography{TC}

\begin{thebibliography}{10}
\providecommand{\url}[1]{\texttt{#1}}
\providecommand{\urlprefix}{URL }

\bibitem{BB88}
Baeten, J.C.M., Bergstra, J.A.: Global renaming operators in concrete process
  algebra. Information and Control  78(3),  205--245 (1988)

\bibitem{BB95a}
Baeten, J.C.M., Bergstra, J.A.: Discrete time process algebra. Formal Aspects
  of Computing  8(2),  188--208 (1996)

\bibitem{BBP13a}
Bergstra, J.A., Bethke, I., Ponse, A.: Cancellation meadows: A generic basis
  theorem and some applications. Computer Journal  56(1),  3--14 (2013)

\bibitem{BK84b}
Bergstra, J.A., Klop, J.W.: Process algebra for synchronous communication.
  Information and Control  60(1--3),  109--137 (1984)

\bibitem{BM06e}
Bergstra, J.A., Middelburg, C.A.: Parallel processes with implicit
  computational capital. Electronic Notes in Theoretical Computer Science  209,
   55--81 (2008)

\bibitem{BPZ07a}
Bergstra, J.A., Ponse, A., van~der Zwaag, M.B.: Tuplix calculus. Scientific
  Annals of Computer Science  18,  35--61 (2008)

\bibitem{BT07a}
Bergstra, J.A., Tucker, J.V.: The rational numbers as an abstract data type.
  Journal of the ACM  54(2),  Article 7 (2007)

\bibitem{CCF90a}
Cissell, R., Cissell, H., Flaspohler, D.: Mathematics of Finance. Houghton
  Mifflin, Boston (1990)

\bibitem{Fei07a}
Fein, M.L.: Financial industry consolidation: The convergence of financial
  products and the implications for regulatory reform. Available at: {\tt
  http://ssrn.com/\linebreak[2]abstract=1654383} (January 2007)

\bibitem{Lex05a}
LexiFi: Structuring, pricing, and processing complex financial products with
  {MLFi}. Available at: {\tt
  http://www.lexifi.com/product/technology/contract-\linebreak[2]description-language}
  (January 2005)

\bibitem{Mes82a}
Meserve, B.E.: Fundamental Concepts of Algebra. Dover Publications, Mineola
  (1982)

\bibitem{ST99a}
Sannella, D., Tarlecki, A.: Algebraic preliminaries. In: Astesiano, E.,
  Kreowski, H.J., Krieg-Br{\"{u}}ckner, B. (eds.) Algebraic Foundations of
  Systems Specification, pp. 13--30. Springer-Verlag, Berlin (1999)

\bibitem{Tar51a}
Tarski, A.: A Decision Method For Elementary Algebra And Geometry. University
  of California Press, Berkeley, second edn. (1951)

\bibitem{WB00a}
Wesseling, J., van~den Bergh, A.: Realistische Interestberekeningen. Academic
  Service, Schoonhoven, the Netherlands (2000)

\bibitem{Wir90a}
Wirsing, M.: Algebraic specification. In: van Leeuwen, J. (ed.) Handbook of
  Theoretical Computer Science, vol.~B, pp. 675--788. Elsevier, Amsterdam
  (1990)

\end{thebibliography}

\begin{comment}
The restriction of a model of \CTC\ to the sort $\Tuplix$, the constant
$\emptx$ and the operator $\conjc$ is a commutative monoid with identity
element $\emptx$.

As mentioned in Example~\ref{example-icap}, there are many different $p$
and $q$ for which
$\cD \models \frac{1}{1+p} \mul (1+q)^2 \leq \frac{8}{5}$.

A wanted financial behaviour may be the buying and selling of an
apartment.

Using a financial product may turn unaffordable financial behaviour
into affordable financial behaviour.

It is intuitively preferable to assume that the alphabets of a wanted
financial behaviour and the financial behaviour of the used 
financial product are disjoint.
\end{comment}

\end{document}